\documentclass[journal]{IEEEtran}
\usepackage[english]{babel}
\usepackage[noadjust]{cite}
\usepackage{amsthm}
\usepackage[linesnumbered,ruled,vlined]{algorithm2e}

\ifCLASSINFOpdf
\else
   \usepackage[dvips]{graphicx}
\fi
\usepackage{url}

\usepackage{multirow}
\usepackage{tabularx}
\usepackage{graphicx}
\usepackage{optidef}
\usepackage{amsmath}
\usepackage{bm}
\usepackage{amssymb}
\usepackage{xcolor}
\usepackage{caption}
\usepackage{subcaption}
\usepackage{lipsum}% http://ctan.org/pkg/lipsum
\usepackage{siunitx}
\usepackage{titlesec}
\titlespacing\section{0pt}{*0.8}{*0.8}
\titlespacing\subsection{0pt}{*0.7}{*0.7}
\titlespacing\subsubsection{0pt}{*0.6}{*0.6}

\newtheorem{lemma}{Lemma}

\begin{document}

\author{
\IEEEauthorblockN{Byunghyun Lee, Anindya Bijoy Das, David J. Love, Christopher G. Brinton, James V. Krogmeier} 

\IEEEauthorblockA{School of Electrical and Computer Engineering, Purdue University, West Lafayette, IN 47907 USA\\
\texttt{\{lee4093,das207,djlove,cgb,jvk\}@purdue.edu}
% \texttt{lee4093@purdue.edu, das207@purdue.edu, djlove@purdue.edu, jvk@purdue.edu, cgb@purdue.edu}
}

\thanks{
% This work is supported in part by the Office of Naval Research under grant N000142112472 and by the National Science Foundation under grants EEC-1941529, CNS-2212565, and CNS-2225577
This work is supported in part by the National Science Foundation under grants EEC-1941529, CNS-2212565, and CNS-2225577 and by the Office of Naval Research under grant N000142112472.}
}

% \title{Constant Envelope Precoding with Block-Level Interference Exploitation for DFRC Systems}
% \title{Constant Modulus Waveform Design with Block-Level Interference Exploitation for Dual-Functional Radar-Communication Systems}
% \title{Block-Level Interference Exploitation Precoding for Dual-Functional Radar-Communication Systems}
\title{Constant Modulus Waveform Design \\ with Block-Level Interference Exploitation for \\ DFRC Systems}

\maketitle

\begin{abstract}
Dual-functional radar-communication (DFRC) is a promising technology where radar and communication functions operate on the same spectrum and hardware.
In this paper, we propose an algorithm for designing constant modulus waveforms for DFRC systems.
Particularly, we jointly optimize the correlation properties and the spatial beam pattern.
For communication, we employ constructive interference-based block-level precoding (CI-BLP) to exploit distortion due to multi-user transmission and radar sensing.
We propose a majorization-minimization (MM)-based solution to the formulated problem.
To accelerate convergence, we propose an improved majorizing function that leverages a novel diagonal matrix structure.
We then evaluate the performance of the proposed algorithm through rigorous simulations.
Simulation results demonstrate the effectiveness of the proposed approach.

\end{abstract}

\begin{IEEEkeywords}
Integrated Sensing and Communication (ISAC), Dual-Functional Radar-Communication (DFRC), Interference Exploitation, Multiple-Input Multiple-Output (MIMO)
\end{IEEEkeywords}

\IEEEpeerreviewmaketitle

\section{Introduction}

To address the increasing spectrum scarcity, the concept of integrated sensing and communication (ISAC) has emerged, aiming to unify radio sensing and communication in a shared spectrum.
The sensing capabilities of communication systems offered by ISAC are expected to play a crucial role in location-based applications such as connected and autonomous vehicles, smart factories, and environmental monitoring.
Initially, ISAC involved information embedding into radar pulses and the coexistence of radar and communication (RadCom).
ISAC technologies then continued to evolve towards dual-functional radar-communication (DFRC), which integrates radar and communication into shared hardware and spectrum \cite{liu2020joint}.
 % \cite{liu2018mu}
% not only shares spectrum but also hardware.
% \cite{liu2022survey,liu2022integrated}
% Earlier forms of ISAC include information embedding with radar pulses and radar-communication (RadCom) coexistence \cite{liu2018mu}.
% \cite{liu2020joint,zhang2021overview}
% Radar sensing often involves high-power transmission for accurate target detection and high-quality parameter estimation.
% Therefore, the design of constant modulus waveforms is vital for maximizing the efficiency of high-power amplifiers (HPAs).

In this context, there is growing research interest in designing waveforms for DFRC systems.
Since DFRC systems often involve high-power transmission for high-quality sensing, designing constant modulus waveforms is essential for the efficiency of high-power amplifiers (HPAs).
Some existing works have tackled the problem of designing constant modulus waveforms for DFRC systems \cite{liu2018toward,liuDualFunctionalRadarCommunicationWaveform2021,liu2022joint}.
% The work in \cite{liu2018toward} minimized the communication MUI under constant modulus and radar waveform similarity constraints.
% In \cite{wu2023constant}, a target SNR maximization problem was solved under a per-user interference constraint.
As an alternative approach, the works in \cite{bazziIntegratedSensingCommunication2023,liu2022joint} considered a peak-to-average power ratio (PAPR) constraint to circumvent the problem of nonlinearity in HPAs.

From a communication perspective, the design of constant modulus waveforms has a strong connection to symbol-level precoding (SLP), which directly designs transmit symbols rather than a linear precoder \cite{li2020tutorial,alodeh2015constructive}.
Specifically, SLP utilizes explicit data symbol information to exploit interference that contributes to communication signal power, called constructive interference (CI).
% as well as channel information, in DFRC, distortion due to radar sensing
Despite extensive research on CI-based SLP (CI-SLP), the use of CI-SLP for DFRC systems has been relatively unexplored.
In \cite{liuDualFunctionalRadarCommunicationWaveform2021}, a beam pattern design problem was tackled under per-user CI and constant modulus constraints.
This work focused on symbol-by-symbol optimization, which requires solving an optimization problem at every symbol time.
 % caused by symbol-by-symbol optimization
To mitigate the computational burden, the work in \cite{li2022practical} studied block-level interference exploitation, also referred to as CI-based block-level precoding (CI-BLP).
In \cite{liu2022joint}, the use of CI-BLP for DFRC systems was initially investigated.
This work followed a cognitive radar framework that utilizes known information about targets and clutter to maximize the radar signal-to-interference-plus-noise ratio (SINR).

% Designing a radar waveform requires careful consideration of its space-time correlation properties to minimize ambiguity in parameter estimation.
% Unlike CI-SLP, using CI-BLP for DFRC systems requires careful consideration of the space-time correlation properties to minimize ambiguity in parameter estimation.
% the correlation properties of the waveform should be carefully addressed to minimize ambiguity in parameter estimation.

Unlike the CI-SLP approach, the CI-BLP approach focuses on optimizing a space-time matrix.
Therefore, from a radar perspective, it becomes crucial to address the correlation properties of the waveform as well as its spatial properties to reduce ambiguity in space and time.
In past ISAC work, a prevalent approach to this challenge has been introducing a similarity constraint \cite{liu2018toward,bazziIntegratedSensingCommunication2023,qian2018joint,keskinLimitedFeedforwardWaveform2021,liu2022joint}.
The similarity constraint ensures that the designed waveform retains space-time correlation properties of the reference waveform, such as linear frequency modulated (LFM) waveforms.
% Nonetheless, this approach offers an indirect solution to space-time sidelobes and thus lacks direct control over them.
Nonetheless, such an indirect approach lacks direct control over space-time sidelobes.
Direct approaches aim to optimize explicit sidelobe cost functions such as integrated sidelobe level (ISL) and peak sidelobe level (PSL) \cite{liuJointTransmitBeamforming2020a,liuRangeSidelobeReduction2020,wen2023transmit}.
% In \cite{liuJointTransmitBeamforming2020a}, the radar beam pattern and spatial cross-correlation were jointly optimized under a minimum SINR constraint for communication users.
% In \cite{liuRangeSidelobeReduction2020}, the authors jointly optimized multi-user interference, range ISL, and radar waveform shaping.
% The work in \cite{wen2023transmit} jointly optimized the spatial ISL ratio and beam pattern shaping cost with a minimum SINR constraint for communication.
These works have addressed the reduction of space and time sidelobes individually.
However, space-time correlation properties should be considered together to separate targets at different angles and distances effectively \cite{duly2013time,san2007mimo}.

In this paper, we address the problem of designing constant modulus waveforms for DFRC systems.
We consider block-level optimization for designing DFRC waveforms instead of symbol-by-symbol optimization.
We jointly optimize the correlation properties and spatial beam pattern of the waveform.
For communication, we employ CI-BLP to take advantage of CI on a block level.
% distortion due to multi-user transmission and radar sensing
% We then formulate a constant waveform design problem and demonstrate that the formulated problem is nonconvex.
We then formulate a constant waveform design problem that optimizes the beam pattern and sidelobes subject to a per-user quality of service (QoS) constraint.
To tackle the nonconvexity of the waveform design problem, we propose a solution algorithm based on the majorization-minimization (MM) principle and the method of Lagrange multipliers.
To improve convergence speed, we propose an improved majorizing function that leverages a novel diagonal matrix structure.
We evaluate the proposed algorithm via comprehensive simulations and demonstrate the effectiveness of our approach.

\section{System Model and Problem Formulation}
% \subsection{System Description}
\subsection{System Setup}
Consider a downlink narrowband DFRC system where the base station (BS) serves as a multi-user multiple-input multiple-output (MU-MIMO) transmitter and colocated MIMO radar simultaneously, as depicted in Fig. \ref{fig:system}.
The BS is equipped with transmit and receive arrays of $N_T$ and $N_R$ antennas, respectively.
% The BS is equipped with $N_T$ transmit antennas and $N_R$ receive antennas, which are identical and colocated.
The primary function of the considered system is radar sensing, while the secondary function is communication.
To accomplish the dual functions of radar and communication, this paper focuses on downlink transmission where the BS transmits a discrete-time waveform matrix $\textbf{X}\in \mathbb{C}^{N_T \times L}$ in each transmission block.
The $(n,\ell)$th entry ${X}_{n,\ell}$ represents the $\ell$th discrete-time transmit symbol and the $\ell$th radar subpulse of the $n$th transmit antenna.

% The considered system aims to optimize radar performance while providing minimum communication rates with the users, i.e., radar-centric design, by designing the matrix $\textbf{X}$.
% Specifically, we consider the beam pattern shaping cost and angular-range sidelobe level for radar design metrics.

\begin{figure}[!t]
\center{\includegraphics[width=.7\linewidth]{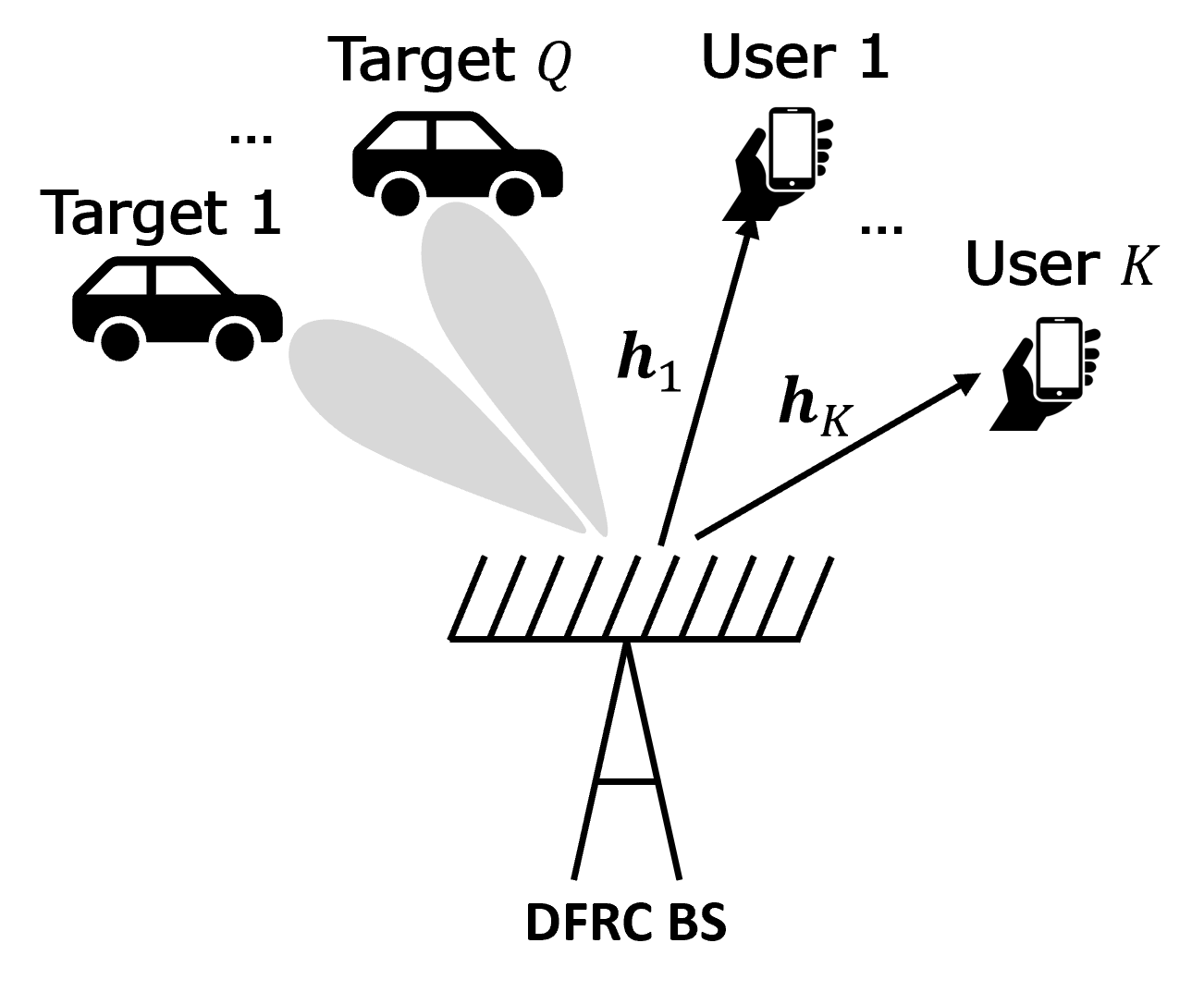}}
\caption{Example of a DFRC system.}
\label{fig:system}
% \vspace{-6mm}
\end{figure}

\subsection{Radar Performance Metrics}

\subsubsection{Beam Pattern Shaping Cost}

In radar waveform synthesis, it is crucial to have strong mainlobes aimed toward targets while maintaining low sidelobes.
This is to ensure strong return signals from the targets and reduce unwanted signals induced by clutter.
Given the waveform $\textbf{X}$, the beam pattern at angle $\theta$ is given by \cite{li2008mimo}
\begin{equation}    \tilde{G}(\textbf{X},\theta)=\textbf{a}^H(\theta)\textbf{X}\textbf{X}^H\textbf{a}(\theta) = \Vert \textbf{a}^H(\theta)\textbf{X} \Vert ^2,
\end{equation}
where 
$\textbf{a}(\theta)\in\mathbb{C}^{N_T}$ is the steering vector of the transmit array.
% $\textbf{a}(\theta)=\frac{1}{\sqrt{N_T}}[1,e^{j\pi\sin\theta},\dots,e^{j\pi(N_T-1)\sin\theta}]\in\mathbb{C}^{N_T}$ is the steering vector of the ULA.
The beam pattern can be expressed in vector form as \vspace{-1mm}
\begin{align*}
    G(\bm{x},\theta) &= \Vert(\textbf{I}_L\otimes\textbf{a}^H(\theta))\bm{x}\Vert^2\\
    &=\bm{x}^H\underbrace{\left(\textbf{I}_L\otimes\textbf{a}^H(\theta))^H(\textbf{I}_L\otimes\textbf{a}^H(\theta)\right)}_{\textbf{A}_u} \bm{x} = \bm{x}^H\textbf{A}_u\bm{x},
    % &= \bm{x}^H\textbf{A}_u\bm{x},
\end{align*}
where $\bm{x}=\text{vec}(\textbf{X})$.
To obtain desired properties, we consider minimizing the means square error (MSE) between the ideal beam pattern and the actual beam pattern, which can be expressed as \vspace{-2mm}
\begin{equation}
    \tilde{g}_{bp}(\alpha,\bm{x})=\displaystyle\sum_{u=1}^U| \alpha G_d(\theta_u)-G(\bm{x},\theta_u)|^2,
\end{equation}
where $U$ is the number of discretized angles, $\alpha$ is a scaling coefficient, and $G_d(\theta_u)$ is the desired beam pattern at angle $\theta_u$.
For convenience, we use a finite number of angles to approximate the beam pattern MSE.
Additionally, the scaling coefficient $\alpha$ adjusts the amplitude of the beam pattern that varies according to the BS transmit power.
% \begin{equation}
%     g_{bp}(\alpha,\bm{x})=\displaystyle\sum_{u=1}^U\left(|\bm{x}^H\textbf{A}_u\bm{x}|^2-\alpha G_d(\theta_u)\bm{x}^H\textbf{A}_u\bm{x}+\alpha^2G_d(\theta_u)^2\right)
% \end{equation}
 % \alpha G_d(\theta_u)-G(\bm{x},\theta_u)
Given that the closed-form solution to $\alpha$ is available, beam pattern shaping cost can be written in compact form as \cite{liuDualFunctionalRadarCommunicationWaveform2021}
(see Appendix \ref{sec:appendix_A} for details) 
\begin{equation}
    g_{bp}(\bm{x})=\displaystyle\sum_{u=1}^U|\bm{x}^H\textbf{B}_u\bm{x}|^2,
\end{equation}
where \vspace{-1.5mm}
\begin{align*}
    % \textbf{B}_u &\triangleq \left(\frac{G_d(\theta_u)\displaystyle\sum_{u'=1}^U\textbf{A}_{u'}G_d(\theta_{u'})}{\displaystyle\sum_{u'=1}^MG_d^2(\theta_{u'})}-\textbf{A}_u\right).
    \textbf{B}_u \triangleq {G_d(\theta_u)\displaystyle\sum_{u'=1}^U\textbf{A}_{u'}G_d(\theta_{u'})}\Big({\displaystyle\sum_{u'=1}^UG_d^2(\theta_{u'})}\Big)^{-1}-\textbf{A}_u.
\end{align*}

\subsubsection{Space-time Autocorrelation and Cross-Correlation ISL}
For designing a radar waveform, its inherent ambiguity should be addressed as it directly impacts the quality of parameter estimation.
We exploit the space-time correlation function to quantify such ambiguity in the radar waveform.
The space-time correlation function characterizes the correlation between a radar waveform and its echo reflected from different points in space and time, which is given by \cite{wang2012design}
\begin{equation}
\chi_{\tau,{q},{q'}}=|\textbf{a}^H(\theta_q)\textbf{X}\textbf{J}_{\tau}\textbf{X}^H\textbf{a}(\theta_{q'})|^2,
\end{equation}
where $\textbf{J}_{\tau}\in\mathbb{R}^{L \times L}$ is the shift matrix.
The shift matrix accounts for the time shifts of the waveform due to the round-trip delay between the BS and a target, which is given by \cite{horn2012matrix}
\begin{equation}
    [\textbf{J}_{\tau}]_{i,j} = \begin{cases}
        1, & \text{if }j-i=\tau \\
        0, & \text{otherwise}.
    \end{cases}
\end{equation}
% with $\chi_{\tau,q,q'}=\chi_{-\tau,q',q}$. 
The space-time correlation function can be rewritten in vector form as (See Appendix \ref{sec:appendix_B} for details)
\begin{align*}
\chi_{\tau,q,q'} = |\bm{x}^H \textbf{D}_{\tau,q,q'}\bm{x}|^2,
\end{align*}
where $\textbf{D}_{\tau,q,q'}=\textbf{J}_{-\tau}\otimes \textbf{a}(\theta_{q'})\textbf{a}^H(\theta_{q})$.
% \vspace{-1.5mm}
% \begin{align*}
    % \textbf{D}_{\tau,q,q'}=\textbf{J}_{-\tau}\otimes \textbf{a}(\theta_{q'})\textbf{a}^H(\theta_{q}).
% \end{align*}
For a given parameter set $(\tau,q,q')$, the space-time correlation function $\chi_{\tau,{q},{q'}}$ describes the correlation between angles $\theta_q$ and $\theta_{q'}$ at a range bin $\tau$.
When $q=q'$, the space-time correlation function represents the autocorrelation properties at angle $\theta_q$.
The autocorrelation integrated sidelobe level (ISL) can be obtained as \vspace{-1mm}
\begin{align}\label{eq:obj_ac}
    g_{ac}(\bm{x})&=\displaystyle\sum_{q=1}^Q\displaystyle\sum_{\substack{\tau=-P+1, \\ \tau\neq 0}}^{P-1}\chi_{\tau,q,q}, \vspace{-1.5mm}
\end{align}
where $P$ is the largest range bin of interest with $P-1\leq L$.
When $q\neq q'$, the space-time correlation function $\chi_{\tau,q,q'}$ represents the cross-correlation properties between angles $\theta_q$ and $\theta_{q'}$ at a range bin $\tau$.
The cross-correlation ISL is given by \vspace{-1mm}
\begin{align}\label{eq:obj_cc}
    g_{cc}(\bm{x})=
    \displaystyle\sum_{q=1}^{Q}\displaystyle\sum_{\substack{q'=1, \\ q'\neq q}}^Q\displaystyle\sum_{\tau=-P+1}^{P-1}\chi_{\tau,q,q'}.
\end{align}
% By minimizing ISL cost functions, we aim to better separate targets at different angles and range bins and estimate angle/range parameters more accurately.

\subsection{Communication Model and QoS Constraint}

Consider MU-MIMO transmission where the BS communicates with $K$ single antenna users simultaneously, i.e., $N_T$ is assumed to be greater than or equal to $K$.
We adopt a block-fading channel model where the communication channels remain the same within a transmission block.
In addition, we assume the BS has perfect knowledge of the user channels $\textbf{h}_k\in \mathbb{C}^{N_T}$ for $k=1,2,\dots,K$.
% In addition, we assume the BS has perfect knowledge of the user channel matrix $\textbf{H}=[\textbf{h}_1,\textbf{h}_2,\dots,\textbf{h}_K]^T\in \mathbb{C}^{K \times N_T}$.
The $\ell$th received symbol at user $k$ can be written as \vspace{-1.5mm}
\begin{equation}
    {y}_{k,\ell}=\textbf{h}_k^H\bm{x}_{\ell}+{n}_{k,\ell},
\end{equation}
where $\bm{x}_{\ell}$ is the $\ell$th column of $\textbf{X}$ and ${n}_{k,\ell}\in \mathbb{C}$ is Gaussian noise with ${n}_{k,\ell} \sim \mathcal{CN}(0,\sigma^2)$.
The codeword for user $k$ is given by $\textbf{s}_k=[s_{k,1},s_{k,2},\dots,s_{k,L}]\in \mathbb{C}^{L}$ where each desired symbol $s_{k,\ell}$ is drawn from a predefined constellation $\mathcal{S}$.
We explain the relationship between $\bm{x}_{\ell}$ and $s_{k,\ell}$ in the following.

\begin{figure}[!t]
\center{\includegraphics[width=.700\linewidth]{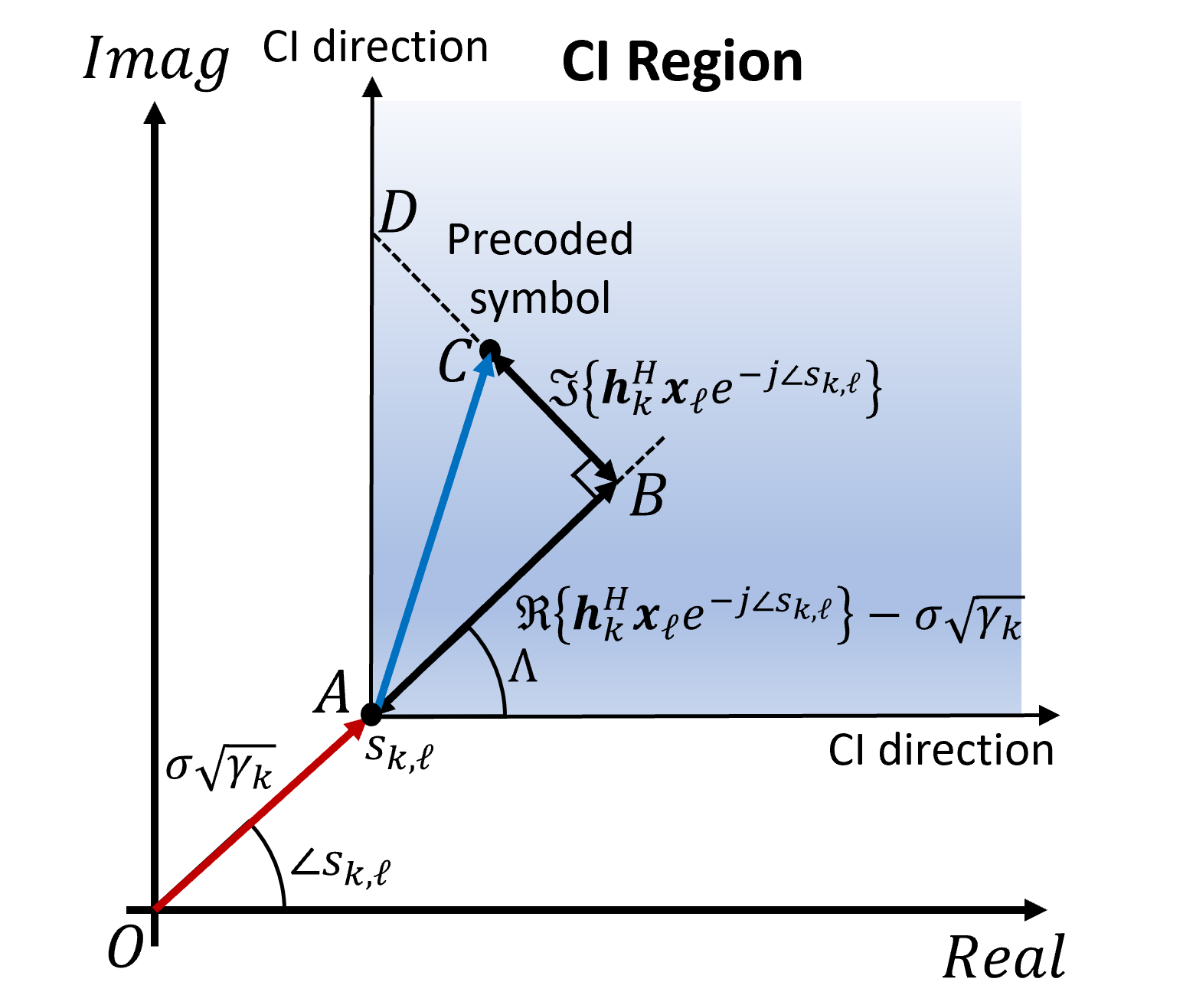}}
\caption{ Constructive interference (CI) region.}
\label{fig:CIR}
% \vspace{-6mm}
\end{figure}

\subsubsection{Per-User Communication QoS Constraint}

% CI refers to an unintended signal that moves the desired symbol $s_{k,\ell}$ farther away from its corresponding decision boundaries in the constructive direction, as illustrated in Fig. \ref{fig:CIR}.

The BS must provide a minimum QoS for the communication users to accomplish the communication task.
We consider a CI-BLP approach to exploit the distortion induced by MU-MIMO and radar transmission.
% Fig. \ref{fig:CIR} describes an example of CI in the quadrature-PSK (QPSK) constellation.
CI refers to an unintended signal that moves the desired symbol $s_{k,\ell}$ farther away from its corresponding decision boundaries in the constructive direction, as illustrated in Fig. \ref{fig:CIR}.
% Instead of enforcing the received signal $\tilde{\textbf{h}}^H\bm{x}$ to be around the desired symbol $s_{k,\ell}$, CI-SLP approach, we . 
Instead of suppressing signal distortion, the CI-BLP approach leverages CI to reduce symbol error rates.

% \footnote{Although this paper focuses on the PSK scenario, it is possible to extend it to quadrature amplitude modulation (QAM). We refer the readers to \cite{li2022practical} for details.}
In this paper, we focus on the M-phase shift keying (M-PSK) constellation.
Without loss of generality, we consider the QPSK as an example, i.e., $M=2$.
Fig. \ref{fig:CIR} describes the CI region for the QPSK scenario, where the vector $\overrightarrow{OC}=\textbf{h}_k^H\bm{x}_{\ell}$ represents the $\ell$th noiseless precoded symbol for user $k$.
A precoded symbol lies within the CI region if the condition $|\overrightarrow{BD}|-|\overrightarrow{BC}|\geq 0$ holds.
The lengths $|\overrightarrow{BD}|$, $|\overrightarrow{BC}|$ of the direction vectors can be expressed, respectively, as \vspace{-1mm}
\begin{align}
    |\overrightarrow{BD}|&=\Re\{\textbf{h}^H_k\bm{x}_{\ell}e^{-j\angle s_{k,\ell}}-\sigma \sqrt{\gamma_k}\}\tan \Lambda \\
    |\overrightarrow{BC}|&=|\Im\{\textbf{h}^H_k\bm{x}_{\ell}e^{-j\angle s_{k,\ell}}\}|.    \vspace{-1.5mm}
\end{align}
$|\overrightarrow{OA}|=\sigma \sqrt{\gamma_k}$ represents the QoS requirement for the users, which can be viewed as a signal-to-noise ratio (SNR) requirement.
From the above discussion, the communication constraint for the $\ell$th symbol of user $k$ can be formulated as 
\begin{align*}\label{ineq:comm_constraint}
    \Re&\{\textbf{h}^H_k\bm{x}_{\ell}e^{-j\angle s_{k,\ell}}-\sigma \sqrt{\gamma_k}\}\tan \Lambda -|\Im\{\textbf{h}^H_k\bm{x}_{\ell}e^{-j\angle s_{k,\ell}}\}| \geq 0,
\end{align*}
where $\Lambda = \pi/M$.
% Due to the limited space, we refer the readers to \cite{liu2022joint} for the detailed derivation of \eqref{ineq:comm_constraint}.
The above CI constraint can be transformed into compact form as \cite{liu2022joint} \vspace{-1mm}
\begin{equation}
    \Re\{\tilde{\textbf{h}}^H_m\bm{x}\} \geq \gamma_m, \forall m=1,2,\dots,2KL
\end{equation}
where \vspace{-1.5mm} 
\begin{align*}    
    \tilde{\textbf{h}}^H_{(2\ell-2)K+k}& \triangleq  \textbf{e}^T_{\ell,L} \otimes \textbf{h}_k^He^{-j\angle s_{k,\ell}}(\sin \Lambda -j\cos \Lambda)\\
    \tilde{\textbf{h}}^H_{(2\ell-1)K+k}& \triangleq  \textbf{e}^T_{\ell,L} \otimes \textbf{h}_k^He^{-j\angle s_{k,\ell}}(\sin \Lambda +j\cos \Lambda)\\
    \Gamma_{(2\ell-2)K+k}  &\triangleq \sigma \sqrt{\gamma_k}\sin \Lambda,\Gamma_{(2\ell-1)K+k} \triangleq \sigma \sqrt{\gamma_k}\sin \Lambda,
\end{align*}
and $\textbf{e}_{\ell,L}\in \mathbb{R}^{L}$ denotes the $\ell$th column of $\textbf{I}_L$.

\subsection{Problem Formulation}
% Our design goal is to optimize radar performance

Based on the formulated performance metrics, our objective is to design a dual-functional waveform for detecting targets at specific angles while also serving the communication users.
To achieve this, we aim to jointly minimize the beam pattern shaping cost, autocorrelation ISL, and cross-correlation ISL.
% To achieve this, we aim to minimize the beam pattern shaping cost, thereby maximizing the beamforming gain for spatial angles $\theta_1,\theta_2,\dots,\theta_Q$ while minimizing spatial sidelobes.
% In addition, we also minimize the space-time ISL to reduce the autocorrelation of each target and the cross-correlation between the targets.
From a communication perspective, we ensure that the communication symbols fall within the CI region to meet the QoS requirement.
By taking these design goals into account, the waveform design problem is formulated as
\begin{equation}\label{eq:prob_formulation1}
\begin{aligned}
& \underset{\bm{x}}{\min}
& &  \omega_{bp}g_{bp}(\bm{x})+\omega_{ac} g_{ac}(\bm{x})+\omega_{cc} g_{cc}(\bm{x})  \\ 
& \text{s.t.}
& & \textbf{C1}: \Re\{\tilde{\textbf{h}}^H_m\bm{x}\} \geq \Gamma_m, \forall m=1,2,\dots,2KL \\
& & & \textbf{C2}: | {x}_{n} | = \sqrt{\frac{P_T}{N_T}}, \ \forall n=1,2,\dots,LN_T \\
\end{aligned}
\end{equation}
where $\omega_{bp}$, $\omega_{ac}$, $\omega_{cc}$ are the weights for the beam pattern shaping cost, autocorrelation ISL, and cross-correlation ISL, respectively.
Note $\textbf{C1}$ is the communication constraint, $\textbf{C2}$ is the constant modulus constraint.

\section{Proposed MM-based Solution}

In this section, we introduce our solution using the MM technique.
To address the intractable fourth-order objective, we derive its linear majorizer.
For faster convergence, we introduce an improved majorizer based on a novel diagonal matrix structure.
% The convergence behavior of MM algorithms is dominated by the choice of the majorizer. 
% Motivated by this, we propose a novel majorizer that outperforms existing majorizers in terms of convergence speed.
By using the proposed majorizer, the original problem can be approximated to linear programming (LP) with a constant modulus constraint.
It is shown that this class of problems can be efficiently tackled using the method of Lagrange multipliers \cite{he2022qcqp}. 
% A solution can be found by iteratively solving approximated problems until convergence.
We demonstrate the majorization process of \eqref{eq:prob_formulation1} and the solution based on dual problems. 

% The approximated problem is tractable, which can be effectively solved by the method of Lagrange multipliers.
% The iMM technique was initially proposed to tackle quadratically constrained quadratic programming (QCQP) with a constant modulus constraint \cite{he2022qcqp}.
% To be specific, the iMM technique combines the majorization of the objective and inner approximation of the feasible set. aiming to find a linear majorizer and inner approximation.
% As a result, the original problem can be approximated to linear programming (LP) with a constant modulus constraint.

% \vspace{-3mm}

\subsection{Majorization via an Improved Majorizer}

% \begin{equation}
%     g_{bp}(\alpha,\bm{x})=\displaystyle\sum_{p=1}^{U}(|\bm{x}^H\textbf{A}_u\bm{x}|^2+\alpha^2G^2_{d}(\theta_u)+2\alpha G^2_{d}(\theta_u)\bm{x}^H\textbf{A}_u\bm{x})
% \end{equation}
 
First, we rewrite the quadratic term in the beam pattern shaping cost as $\bm{x}^H\textbf{B}_u\bm{x}=\text{Tr}(\bm{x}\bm{x}^H\textbf{B}_u)=\text{vec}^H(\bm{x}\bm{x}^H)\text{vec}(\textbf{B}_u)$.
Then, adopting the common approach in \cite{liFastAlgorithmsDesigning2018,liuDualFunctionalRadarCommunicationWaveform2021,liFastAlgorithmsDesigning2018,sun2016majorization}, the fourth-order beam pattern shaping cost can be expressed as
\begin{equation}\label{eq:obj_quartic}
\begin{aligned}
    g_{bp}(\bm{x})&=\displaystyle\sum_{p=1}^{U}|\bm{x}^H\textbf{B}_u\bm{x}|^2  \\ \nonumber
    % &=\displaystyle\sum_{p=1}^{U}\text{vec}^H(\bm{x}\bm{x}^H)\text{vec}(\textbf{B}_u)\text{vec}^H(\textbf{B}_u)\text{vec}(\bm{x}\bm{x}^H) \\ \nonumber
    &= \text{vec}^H(\bm{x}\bm{x}^H)\underbrace{\left(\displaystyle\sum_{u=1}^{U}\text{vec}(\textbf{B}_u)\text{vec}^H(\textbf{B}_u)\right)}_{{\bm{\Psi}}_1}\text{vec}(\bm{x}\bm{x}^H) \\ \nonumber
    &= \text{vec}^H(\bm{x}\bm{x}^H){\bm{\Psi}}_1\text{vec}(\bm{x}\bm{x}^H).
\end{aligned}
\end{equation}
It can be verified that ${\bm{\Psi}}_1$ is a $(L^2N_T^2 \times L^2N_T^2)$ Hermitian positive definite matrix.
Following this approach, the objective can be expressed as
\begin{equation}\label{eq:obj_major1}
    \begin{aligned}
g(\bm{x})&=\text{vec}^H(\bm{x}\bm{x}^H)\underbrace{\left(\omega_{bp}\bm{\Psi}_1+\omega_{ac}\bm{\Psi}_2+\omega_{cc}\bm{\Psi}_3\right)}_{\bm{\Psi}}\text{vec}(\bm{x}\bm{x}^H) \\ 
    &=\text{vec}^H(\bm{x}\bm{x}^H)\bm{\Psi}\text{vec}(\bm{x}\bm{x}^H)
\end{aligned}
\end{equation}
where \vspace{-1.5mm}
\begin{align*}
    \bm{\Psi}_2 &\triangleq \displaystyle\sum_{q=1}^Q\displaystyle\sum_{\substack{\tau=-P+1, \\ \tau\neq 0}}^{P-1}\text{vec}(\textbf{D}_{\tau,q,q})\text{vec}^H(\textbf{D}_{\tau,q,q}) \\
    \bm{\Psi}_3 &\triangleq\displaystyle\sum_{q=1}^{Q}\displaystyle\sum_{\substack{q'=1, \\ q'\neq q}}^Q\displaystyle\sum_{\tau=-P+1}^{P-1}\text{vec}(\textbf{D}_{\tau,q,q'})\text{vec}^H(\textbf{D}_{\tau,q,q'}).
\end{align*}

Then, we use the following lemma to construct a majorizer of the fourth-order objective function.
\begin{lemma}\label{lemma:Taylor} (\cite[(13)]{sun2016majorization})
Let $\textbf{Q},\textbf{R}$ be Hermitian matrices with $\textbf{R}\succeq \textbf{Q}$. 
Then, a quadratic function $\bm{u}^H\textbf{Q}\bm{u}$ can be majorized at a point $\bm{u}_t$ as
\begin{align*}
    \bm{u}^H\textbf{Q}\bm{u} \leq \bm{u}^H\textbf{R}\bm{u}&+2\Re\{\bm{u}^H(\textbf{Q}-\textbf{R})\bm{u}_t\} +\bm{u}_t^H(\textbf{R}-\textbf{Q})\bm{u}_t.
\end{align*}
\end{lemma}
% \begin{proof}
%     It follows from $(\bm{u}-\bm{u}_t)^H(\textbf{R}-\textbf{Q})(\bm{u}-\bm{u}_t) \geq 0$.
%     The proof is complete.
% \end{proof}
By choosing a diagonal $\textbf{R}$ such that $\textbf{R}\succeq \textbf{Q}$, one can majorize a quadratic function.
In the literature, the most widely used choice for $\textbf{R}$ is $\textbf{R}=\lambda_Q \textbf{I}$ where $\lambda_Q$ is the largest eigenvalue of $\textbf{Q}$ \cite{liuDualFunctionalRadarCommunicationWaveform2021,zhao2016unified,sun2016majorization}.
However, this majorizer can be loose when $\textbf{Q}$ is ill-conditioned, which may incur slow convergence.
% In \cite{song2015sequence}, a novel diagonal matrix structure for $\textbf{R}$ was proposed for the case where $\textbf{Q}$ is a non-negative symmetric matrix. 
% This work showed that a tighter majorizer can improve convergence speed significantly.
% for a quadratic function with a complex Hermitian matrix
To overcome this problem, we propose a novel majorizer based on the following lemma.
\begin{lemma}\label{lemma:diag}
Let $\textbf{Q}$ be a Hermitian matrix.
Let $\bar{\textbf{Q}}$ be a matrix such that $\bar{{Q}}_{i,j}=|{Q}_{i,j}|$.
Then, $\text{diag}(\bar{\textbf{Q}}\textbf{1})\succeq \textbf{Q}$.
\end{lemma}
\begin{proof}
For any $\bm{u}$, we have
\begin{equation}
    \begin{aligned}
    &\bm{u}^H(\text{diag}(\bar{\textbf{Q}}\textbf{1})-\textbf{Q})\bm{u}
    % = \displaystyle\sum_{i,j}|u_i|^2|{Q}_{i,j}|-\displaystyle\sum_i\displaystyle\sum_{j}u_i^*{Q}_{i,j}u_j
    \\ \nonumber
    &= \displaystyle\sum_i\underbrace{u^*_iu_i}_{|u_i|^2}\displaystyle\sum_{j}|{Q}_{i,j}|-\displaystyle\sum_i\displaystyle\sum_{j}u_i^*{Q}_{i,j}u_j
    \\ \nonumber
    &= 
    \displaystyle\sum_{i,j}\left(|{Q}_{i,j}||u_i|^2 
    - {Q}_{i,j}u_i^*u_j\right)
    \\ \nonumber
     & =
    \frac{1}{2}\displaystyle\sum_{i,j}\left(|{Q}_{i,j}||u_i|^2 +|{Q}_{i,j}||u_j|^2 
    - 2\Re\{{Q}_{i,j}u_i^*u_j\}\right).\vspace{-1.5mm}
\end{aligned}
\end{equation}
For any $i,j$, we have
\begin{align*}
    &|{Q}_{i,j}||u_i|^2 +|{Q}_{i,j}||u_j|^2 
    - 2\Re\{{Q}_{i,j}u_i^*u_j\} \geq \\ \nonumber
    &|{Q}_{i,j}||u_i|^2 +|{Q}_{i,j}||u_j|^2 
    - 2|{Q}_{i,j}||u_i||u_j|  \\ \nonumber
    &=|{Q}_{i,j}|(|u_i|- |u_j|)^2 \geq 0,
\end{align*}
% where the first inequality follows from the fact $|{Q}_{i,j}|\Re\{u_i^*u_j\}\geq \Re\{{Q}_{i,j}u_i^*u_j\}$.
where the first inequality follows from the fact $|{Q}_{i,j}||u_i||u_j|\geq \Re\{{Q}_{i,j}u_i^*u_j\}$.
It follows that $\bm{u}^H(\text{diag}(\bar{\textbf{Q}}\textbf{1})-\textbf{Q})\bm{u} \geq 0 \ \forall \bm{u}$.
\end{proof}

Using Lemma 3, a tight majorizer for the objective can be constructed as follows.
% Following these approaches, the beam pattern shaping cost \eqref{eq:obj_quartic} can be majorized as follows.
\begin{lemma}
    % Let $\lambda_{\bm{\Psi}_1}$ be the largest eigenvalue of $\bm{\Psi}_1$.
    Let $\bar{\bm{\Psi}}$ be a matrix such that $\bar{{\Psi}}_{i,j}=|{{\Psi}}_{i,j}|$ for all $i,j$.
    The objective \eqref{eq:obj_major1} can be majorized as
    \begin{equation}
        g_{bp}(\bm{x})\leq \bm{x}^H\bm{\Phi}\bm{x} + const
    \end{equation}
    where \vspace{-1.5mm}
 \begin{align*}    
  {\bm{\Phi}}&\triangleq 2\left(\omega_{bp}\bm{\Phi}_1 + \omega_{ac}\bm{\Phi}_2 +\omega_{cc}\bm{\Phi}_3
-\left({\textbf{E}}\odot\bm{x}_t\bm{x}_t^H\right)\right) \\
\bm{\Phi}_1 & \triangleq \displaystyle\sum_{u=1}^U\bm{x}_t^H{\textbf{B}}_u^H\bm{x}_t{\textbf{B}}_u, 
\bm{\Phi}_2  \triangleq \displaystyle\sum_{q=1}^Q\displaystyle\sum_{\substack{\tau=-P+1, \\ \tau\neq 0}}^{P-1}\bm{x}_t^H\textbf{D}^H_{\tau,q,q}\bm{x}_t\textbf{D}_{\tau,q,q}
\\
\bm{\Phi}_3 & \triangleq \displaystyle\sum_{q=1}^{Q}\displaystyle\sum_{\substack{q'=1, \\ q'\neq q}}^Q\displaystyle\sum_{\tau=-P+1}^{P-1}\bm{x}_t^H\textbf{D}^H_{\tau,q,q'}\bm{x}_t\textbf{D}_{\tau,q,q'},\
{\textbf{E}}\triangleq\text{mat}(\bar{\bm{\Psi}}\textbf{1})
% {\bm{\Phi}}&\triangleq \omega_{bp}\bm{\Phi}_1 + \omega_{ac}\bm{\Phi}_2 +\omega_{cc}\bm{\Phi}_3
% -2\left({\textbf{E}}\odot\bm{x}_t\bm{x}_t^H\right) \\
% \bm{\Phi}_1 & \triangleq 2\displaystyle\sum_{u=1}^U\bm{x}_t^H{\textbf{B}}_u^H\bm{x}_t{\textbf{B}}_u, \\
% \bm{\Phi}_2 & \triangleq 2\displaystyle\sum_{q=1}^Q\displaystyle\sum_{\substack{\tau=-P+1, \\ \tau\neq 0}}^{P-1}\bm{x}_t^H\textbf{D}^H_{\tau,q,q}\bm{x}_t\textbf{D}_{\tau,q,q}
% \\
% \bm{\Phi}_3 & \triangleq 2\displaystyle\sum_{q=1}^{Q}\displaystyle\sum_{\substack{q'=1, \\ q'\neq q}}^Q\displaystyle\sum_{\tau=-P+1}^{P-1}\bm{x}_t^H\textbf{D}^H_{\tau,q,q'}\bm{x}_t\textbf{D}_{\tau,q,q'},\
% {\textbf{E}}\triangleq\text{mat}(\bar{\bm{\Psi}}\textbf{1})
\end{align*}
\end{lemma}
\textit{Proof:} See Appendix \ref{sec:appendix_C}.

This majorizer is still quadratic, which is difficult to solve under the constant modulus constraint. 
Thus, we further majorize the obtained quadratic function to lower its order as follows.
% \begin{lemma}
%     % Let $\lambda_{\bm{\Phi}}$ be the largest eigenvalue of $\bm{\Phi}$. 
%     Let $\bar{\bm{\Phi}}$ be a matrix such that $\bar{\bm{\Phi}}_{i,j}=|{\bm{\Phi}}_{i,j}|$ for any $i,j$.
%     The quadratic function in \eqref{ineq:quadratic_majorizer3} is majorized by
%     \begin{equation}
%        \bm{x}^H\bm{\Phi}\bm{x} \leq \underbrace{\Re\{\bm{x}^H\textbf{d}\}}_{\bar{g}(\bm{x})}+const,
%     \end{equation}
% where 
% \begin{equation}
%     % \textbf{d}\triangleq2({\bm{\Phi}}-\lambda_{\bm{\Phi}}\textbf{I})\bm{x}_t.
%     \textbf{d}\triangleq2({\bm{\Phi}}-\text{diag}(\bar{\bm{\Phi}}\textbf{1}))\bm{x}_t.
% \end{equation}    
% \end{lemma}
\begin{lemma}
    % Let $\lambda_{\bm{\Phi}}$ be the largest eigenvalue of $\bm{\Phi}$. 
    Let $\bar{\bm{\Phi}}$ be a matrix such that $\bar{{\Phi}}_{i,j}=|{{\Phi}}_{i,j}|$ for any $i,j$.
    The quadratic function in \eqref{ineq:quadratic_majorizer3} is majorized by
    \begin{equation}
       \bm{x}^H\bm{\Phi}\bm{x} \leq \underbrace{\Re\{\bm{x}^H\textbf{d}\}}_{\bar{g}(\bm{x})}+const,
    \end{equation}
where $\textbf{d}\triangleq2({\bm{\Phi}}-\text{diag}(\bar{\bm{\Phi}}\textbf{1}))\bm{x}_t$.   
\end{lemma}
\begin{proof}
% Let $\lambda_{\bm{\Phi}}$ be the largest eigenvalue of $\bm{\Phi}$.
By applying Lemma 2 and Lemma 3 again, we have
    \begin{align*}
    \bm{x}^H{\bm{\Phi}}\bm{x} 
     &\leq \underbrace{\bm{x}^H\text{diag}(\bar{\bm{\Phi}}\textbf{1})\bm{x}}_{\frac{P_T}{N_T}\textbf{1}^T\bar{\bm{\Phi}}\textbf{1}}+
     \Re\{\bm{x}^H\underbrace{2({\bm{\Phi}}-\text{diag}(\bar{\bm{\Phi}}\textbf{1}))\bm{x}_t}_{\textbf{d}
    }\}  \\ \nonumber
    &\quad+\bm{x}_t^H(\text{diag}(\bar{\bm{\Phi}}\textbf{1})-{\bm{\Phi}})\bm{x}_t
    = \Re\{\bm{x}^H\textbf{d}\}+const.
    % 2\Re\{\bm{x}^H({\bm{\Phi}}-\text{diag}(\bar{\bm{\Phi}}\textbf{1}))\bm{x}_t\} 
    %  &\leq \lambda_{{\bm{\Phi}}}\bm{x}^H\bm{x}+2\Re\{\bm{x}^H({\bm{\Phi}}-\lambda_{{\bm{\Phi}}}\textbf{I})\bm{x}_t\} 
    % +\bm{x}_t^H(\lambda_{\bm{\Phi}}\textbf{I}-{\bm{\Phi}})\bm{x}_t
    % \\ \nonumber
    % &\leq \Re\{\bm{x}^H\underbrace{2({\bm{\Phi}}-\lambda_{\bm{\Phi}}\textbf{I})\bm{x}_t}_{\textbf{d}
    % }\} 
    % + \underbrace{\lambda_{\bm{\Phi}} LN_T+\bm{x}_t^H(\lambda_{\bm{\Phi}}\textbf{I}-{\bm{\Phi}})\bm{x}_t}_{const}
    % \\ \nonumber
    % &= \Re\{\bm{x}^H\textbf{d}\}+const.
\end{align*}
The proof is complete.
\end{proof}
Using Lemmas 3 and 4, the objective function can be majorized as
    \begin{equation}\label{eq:majorizer}
       \omega_{bp}g_{bp}(\bm{x})+\omega_{ac} g_{ac}(\bm{x})+\omega_{cc} g_{cc}(\bm{x}) \leq \bar{g}(\bm{x})+const,
    \end{equation}    
    where $\bar{g}(\bm{x}) = \Re\{\bm{x}^H\textbf{d}\}$.

% \begin{theorem}

% Given the constant modulus constraint, the objective function can be majorized as
%     \begin{equation}\label{eq:majorizer}
%        \omega_{bp}g_{bp}(\bm{x})+\omega_{ac} g_{ac}(\bm{x})+\omega_{cc} g_{cc}(\bm{x}) \leq \bar{g}(\bm{x})+const,
%     \end{equation}    
%     where $\bar{g}(\bm{x}) = \Re\{\bm{x}^H\textbf{d}\}$.
% \end{theorem}
% \begin{proof}
%     It follows from Lemma 3 and Lemma 4.
% \end{proof}

\subsection{Solution via the Method of Lagrange Multipliers}

\begin{algorithm}
\DontPrintSemicolon
\SetNlSty{textbf}{\{}{\}}
\SetAlgoNlRelativeSize{0}
\SetNlSty{}{}{}

\caption{Bisection Method\label{alg:K-Bisection}}

\textbf{Input:} Lagrange multiplier vector {$\bm{\nu}$, stopping thresholds $\epsilon_1$, $\epsilon_2$}

\textbf{Initialization:} $i=0$; $\bm{\nu}[0] = \bm{\nu}$, $\hat{g}[0]=\infty$;
% if $\bm{\nu}$ with finite entries exists, $\bm{\nu}[0] = \bm{\nu}$. Else $\bm{\nu}[0]=\textbf{0}_{NL \times 1}$.
With slight abuse of notation, $\bar{g}_m(\nu')$ denotes $\bar{g}_m(\bm{x}(\bm{\nu}))|_{\nu_{m}=\nu'}$\;

% \While{$|\hat{g}[i]-\hat{g}[i-1]|/|\hat{g}[i-1]|>\epsilon_1$}
% \While{$|\Delta\hat{g}[i]|>\epsilon_1$}
% \DoWhile{$|\hat{g}[i]-\hat{g}[i-1]|/|\hat{g}[i-1]|>\epsilon_1$}
\Repeat{$|\hat{g}[i]-\hat{g}[i-1]|/|\hat{g}[i-1]|<\epsilon_1$}{
    \For{$m=1:2KL$}{
        % \eIf{$\bar{g}_m(0)\leq 0$}{
            % $\nu_m = 0$;
        % }
        \lIf{$\bar{g}_m(0)\leq 0$}{$\nu_m^u=0$}
        \Else{
            $\nu^l_m=0$, $\nu^u_m=1$;\;
            \lIf{$\bar{g}_m(\nu^u_m)\leq 0$}{$\nu_m^u=1$}
            \Else{
                % \lWhile{$\bar{g}_m(\nu^u_m)\geq0$}{$\nu^u_m=2\nu_m^u$}
                % $\nu^l_m=\nu^u_m/2$\;
                \lRepeat{$\bar{g}_m(\nu^u_m)\leq0$}{$\nu^u_m=2\nu_m^u$}
                $\nu^l_m=\nu^u_m/2$\;
            }
            % \While{$|\bar{g}_m(\nu_m)+\epsilon_2/2|>\epsilon_2/2$}
            \Repeat{$|\bar{g}_m(\nu_m)+\epsilon_2/2|<\epsilon_2/2$}{
                $\nu_m=(\nu^l_m+\nu^u_m)/2$;\;
                \lIf{$\bar{g}_m(\nu_m)>0$}{$\nu^l_m=\nu_m$}
                \lElse{$\nu^u_m=\nu_m$}
            }
        }
    }
    $i \gets i+1$, $\bm{\nu}[i]=[\nu_1,\dots,\nu_{2KL}]$, $\hat{g}[i]=\bar{g}(\bm{\nu}[i])+\displaystyle\sum_{m=1}^{2KL}\nu_m\bar{g}_m(\bm{\nu}[i])$\;
}

{\textbf{Output:} recover a solution $\bm{x}$ from $\bm{\nu}[i]$ and \eqref{eq:sol_to_x}}

\end{algorithm}

\begin{algorithm}
\DontPrintSemicolon
\SetNlSty{textbf}{\{}{\}}
\SetAlgoNlRelativeSize{0}
\SetNlSty{}{}{}

\caption{Proposed MM-based Algorithm\label{alg:MM}}

{\textbf{Input:} starting point $\bm{x}_0$, stopping threshold $\epsilon_3$}
% Choose a feasible point $\bm{x}_0$, 

\textbf{Initialize:} set $t=0$, $g[0]=\infty$\;

% \While{$|g[t]-g[t-1]|/|g[t-1]|\geq \epsilon_3$}
\Repeat{$|g[t]-g[t-1]|/|g[t-1]|\leq \epsilon_3$}{
    $t \gets t+1$, Update $\bm{x}_t$ using Algorithm \ref{alg:K-Bisection}\;
        % Update $\bm{\nu}$ using Algorithm \ref{alg:K-Bisection}, $\bm{x}_t \gets \bm{x}(\bm{\nu})$.\;
    % \;
    $g[t]\gets\omega_{bp}g_{bp}(\bm{x}_t)+\omega_{ac} g_{ac}(\bm{x}_t)+\omega_{cc} g_{cc}(\bm{x}_t)$\;
    
}

{\textbf{Output:} convert $\bm{x}_t$ into a matrix $\textbf{X}$}

\end{algorithm}

Using the majorization \eqref{eq:majorizer}, the problem \eqref{eq:prob_formulation1} can be approximated as
\begin{equation}\label{eq:iMM_formulation}
\begin{aligned}
& \underset{\bm{x}}{\min}
& &  \bar{g}(\bm{x})  \\ 
& \text{s.t.}
& &  \bar{g}_m(\bm{x})\leq 0, \ \forall m=1,2,\dots,2KL\\
& & & | {x}_{n} | = \sqrt{\frac{P_T}{N_T}}, \ \forall n=1,2,\dots,LN_T \\
\end{aligned}
\end{equation}
where $\bar{g}_m(\bm{x})=\Re\{-\bm{x}^H\tilde{\textbf{h}}_m\}+\Gamma_m$.
% The objective of the above problem is an upper bound of \eqref{eq:prob_formulation2}.
The Lagrange dual problem for \eqref{eq:iMM_formulation} is given by
\begin{equation}\label{eq:iMM_dual}
\begin{aligned}
& \underset{\bm{\nu}}{\sup}\ \underset{\bm{x}}{\min}
& &  \bar{g}(\bm{x}) + \displaystyle\sum_{m=1}^{2KL} \nu_m \bar{g}_m(\bm{x})\\ 
& \text{s.t.}
& &  | {x}_{n} | = \sqrt{\frac{P_T}{N_T}}, \ \forall n=1,2,\dots,LN_T \\
& & & \nu_m \geq 0 , \ \forall m=1,2,\dots,2KL 
\end{aligned}
\end{equation}
where $\bm{\nu}=[\nu_1,\nu_2,\dots,\nu_{2KL}]$ is the Lagrange multiplier vector with ${\nu}_m\in\mathbb{R}^+$ being the Lagrange multiplier for the $m$th communication constraint.
For a given $\bm{\nu}$, the objective of \eqref{eq:iMM_dual} is linear.
Thus, the optimal dual solution can be obtained as \vspace{-1mm}
\begin{equation}\label{eq:sol_to_x}
    \bm{x}^*(\bm{\nu})=\sqrt{\frac{P_T}{N_T}}\exp{\left(j \angle\left(\displaystyle\sum_{m=1}^{2KL}\nu_m\tilde{\textbf{h}}_m-\textbf{d}\right)\right)}.
\end{equation}

It is shown that the strong duality between the primal and dual problems holds \cite{he2022qcqp} if there exists a solution $\bm{\nu}$ that satisfies the following conditions:
\begin{align}
    &\bm{x}(\bm{\nu})=\sqrt{\frac{P_T}{N_T}}\exp{\left(j \angle\left(\displaystyle\sum_{m=1}^{2KL}\nu_m\tilde{\textbf{h}}_m-\textbf{d}\right)\right)}, \label{eq:x_nu}\\
    &0\leq \nu_m \leq \infty,\ \bar{g}_m(\bm{x}(\bm{\nu}))\leq 0, \forall m =1,2,\dots,2KL \label{ineq:dual_feasible}\\
    &\nu_m\bar{g}_m(\bm{x}(\bm{\nu}))=0, \forall m =1,2,\dots,2KL \label{eq:KKT}.
\end{align}

% A solution satisfying \eqref{eq:x_nu} and \eqref{eq:KKT} always exists, given $\nu_m<\infty$ for $m=1,2,\dots,2KL$.
% Assuming that the feasible set is strictly feasible, we have $\lim_{\nu_m\rightarrow\infty}\bar{g}_m(\bm{x}(\bm{\nu}))=\bar{g}_m\left(\exp{\left(j \angle\tilde{\textbf{h}}_m\right)}\right)<0$ for any $m=1,2,\dots,2KL$.
A solution satisfying \eqref{eq:x_nu} and \eqref{eq:KKT} always exists, given $\nu_m<\infty$ for $m=1,2,\dots,2KL$.
Assuming that the feasible set is strictly feasible, we have $\lim_{\nu_m\rightarrow\infty}\bar{g}_m(\bm{x}(\bm{\nu}))=\bar{g}_m\left(\exp{\left(j \angle\tilde{\textbf{h}}_m\right)}\right)<0$ for any $m=1,2,\dots,2KL$.
Hence, there exists $\bm{\nu}$ that satisfies equation \eqref{eq:KKT} and has finite entries, leading to strong duality.
Thus, we focus on solving the dual problem instead of solving the primal problem directly.
Given the closed-form solution \eqref{eq:sol_to_x} to the inner problem, the dual problem \eqref{eq:iMM_dual} can be reduced to finding optimal $\bm{\nu}$ that satisfy the conditions \eqref{eq:x_nu} and \eqref{eq:KKT}.
With this in mind, the dual problem can be reformulated as
\begin{equation}\label{eq:iMM_dual2}
\begin{aligned}
& \underset{\bm{\nu}}{\sup}\ 
& &\bar{g}(\bm{x}(\bm{\nu})) + \displaystyle\sum_{m=1}^{2KL} \nu_m \bar{g}_m(\bm{x}(\bm{\nu}))\\ 
& \text{s.t.} & & \nu_m \geq0 ,\ \bar{g}_m(\bm{x}(\bm{\nu}))\leq 0, \ \forall m=1,2,\dots,2KL \\
& & &\nu_m\bar{g}_m(\bm{x}(\bm{\nu}))=0, \forall m =1,2,\dots,2KL
\end{aligned}
\end{equation}

The problem \eqref{eq:iMM_dual2} can be solved via a coordinate ascent method where one Lagrange multiplier is optimized at a time with the other Lagrange multipliers fixed.
Specifically, we modify the bisection algorithm in \cite{he2022qcqp}, as described in Algorithm \ref{alg:K-Bisection}.
% To solve \eqref{eq:iMM_dual2}, we modify the bisection algorithm in \cite{he2022qcqp}, as described in Algorithm \ref{alg:K-Bisection}.
Once the Lagrange multipliers are obtained, the solution to the primal problem can be directly recovered using \eqref{eq:sol_to_x}.
This process is repeated until the objective value converges, as described in Algorithm \ref{alg:MM}.

% \subsection{Complexity Analysis}
\color{black}

% \color{blue}
\section{Simulation Results}

\begin{figure}[t]
         \centering
        \center{\includegraphics[width=.85\linewidth]
        % {./Fig/fig_Beam_Pattern.eps}}
        {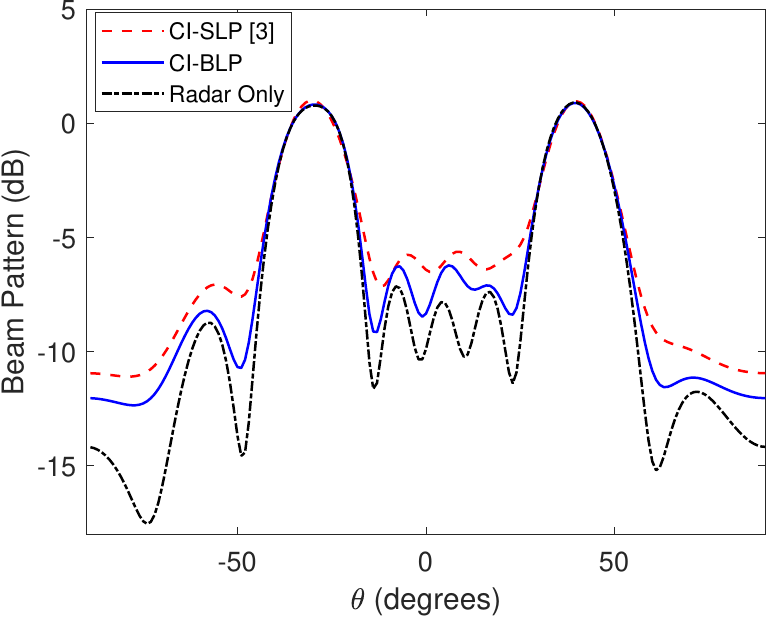}}
        \caption{Synthesized beam patterns. The proposed approach (CI-BLP) achieved lower sidelobes than the CI-SLP approach \cite{liuDualFunctionalRadarCommunicationWaveform2021} owing to block-level optimization.}
        \label{fig:Beam_Pattern1}
        \vspace{-2mm}
\end{figure}

\begin{figure*}
     \centering
     \begin{subfigure}[b]{0.3\textwidth}
            {\includegraphics[width=.95\linewidth]
            % {./Fig/fig_AC1.eps}}
            {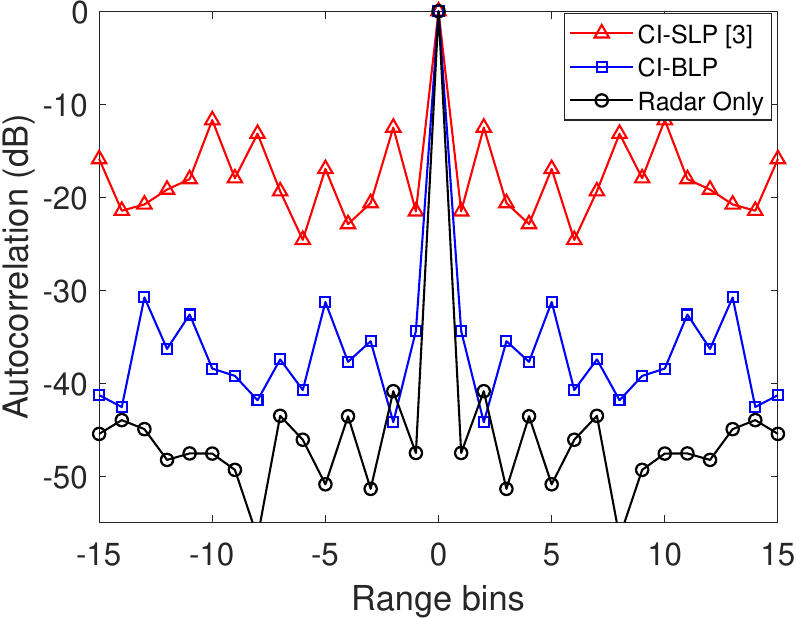}}
            \caption{ Autocorrelation for $\theta_1=-30^{\circ}$.}
            \label{fig:AC1}
     \end{subfigure}
     \begin{subfigure}[b]{0.3\textwidth}
             {\includegraphics[width=.95\linewidth]
             % {./Fig/fig_AC2.eps}}
             {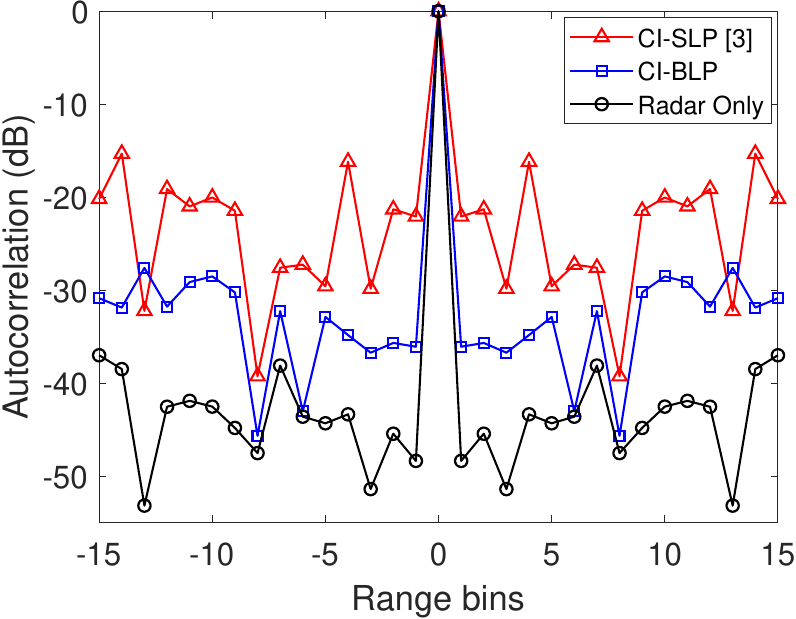}}
             \caption{ Autocorrelation for $\theta_2=40^{\circ}$.}
             \label{fig:AC2}
     \end{subfigure}     
     \label{fig:AC}
     \begin{subfigure}[b]{0.3\textwidth}
            {\includegraphics[width=.95\linewidth]
            % {./Fig/fig_CC1.eps}}
            {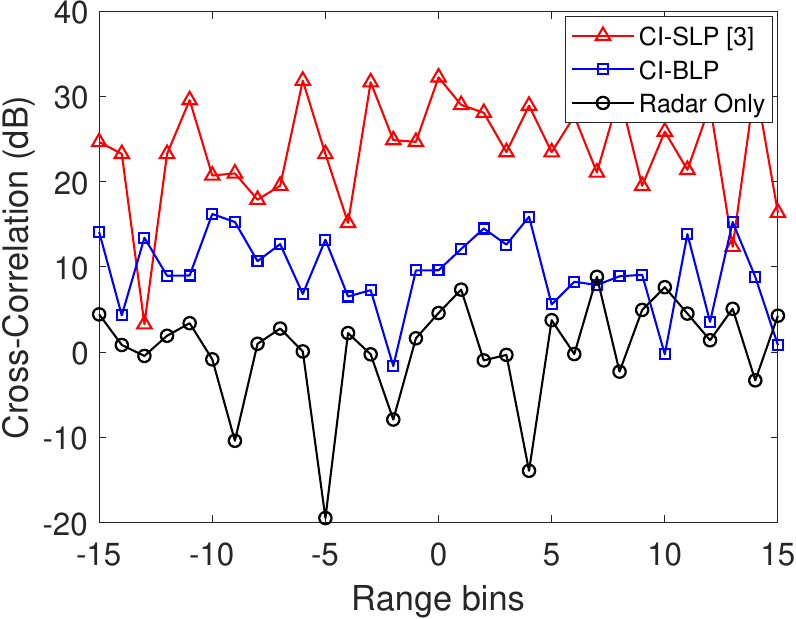}}
            \caption{Cross-correlation between $\theta_1$ and $\theta_2$.}
            \label{fig:CC1}
     \end{subfigure}
     \caption{Space-time autocorrelation and cross-correlation. The proposed CI-BLP approach achieved lower autocorrelation and cross-correlation sidelobes than CI-SLP due to ISL minimization.}
     \vspace{-5mm}
\end{figure*}

% \begin{figure}[t]
%      \centering
%      \begin{subfigure}[b]{0.45\textwidth}
%             \includegraphics[width=.900\linewidth]{./Fig/fig_AC1.eps}
%             \caption{ Autocorrelation for $\theta_1=-30^{\circ}$.}
%             \label{fig:AC1}
%      \end{subfigure}
%      \\
%      \begin{subfigure}[b]{0.45\textwidth}
%              \includegraphics[width=.900\linewidth]{./Fig/fig_AC2.eps}
%              \caption{ Autocorrelation for $\theta_2=40^{\circ}$.}
%              \label{fig:AC2}
%      \end{subfigure}     
%      \caption{Space-time autocorrelations.}
%      \label{fig:AC}
% \end{figure}

In this section, we evaluate the proposed algorithm through simulations.
We set $K=3$, $N_T=10$, $L=64$, $P=16$, $\gamma_k=6\ dB$, $P_T=1$, $\sigma^2=0.01$, $\epsilon_1=\epsilon_2 = 10^{-4}$ and $\epsilon_3=3\times 10^{-5}$ unless otherwise specified.
% Note that $P_T/\sigma^2=10$, i.e., the transmit SNR is set to $10\ dB$.
% The number of users and the SNR threshold are set to $K=3$ and $\gamma_k=6\ dB$, respectively.
For the transmit array, we use a uniform linear array (ULA) with half-wavelength spacing.
We consider the uncorrelated Rayleigh channel for the communication channel of each user, i.e., $\textbf{h}_k\sim \mathcal{CN}(\textbf{0}_{N_T \times 1},\textbf{I}_{N_T})$ for $k=1,2,\dots,K$.
We set the discretized angle range to be $[0^{\circ},180^{\circ}]$ with the angle resolution of $1^{\circ}$, i.e., $\theta_u=u^{\circ}$ for $u=1,2,\dots,180$. 
For the reference beam pattern, we consider a rectangular beam pattern, which is given by \cite{li2008mimo}
\begin{equation}
    G_d(\theta)=
    \begin{cases}
    1,& \text{if } \theta_q - \Delta_{\theta}/2\leq\theta\leq  \theta_q + \Delta_{\theta}/2 \ \forall q\\
    0,              & \text{otherwise}
\end{cases}
\end{equation}
where $\Delta_\theta$ is the beam width.
We considered two targets, i.e, $Q=2$ each at angle $\theta_1=-30^{\circ}$ and $\theta_3=40^{\circ}$.
The beam width $\Delta_{\theta}$ is set to $20^{\circ}$.
We set the weights for the cost functions as $(\omega_{bp},\omega_{ac},\omega_{cc})=(1,2,2)$.
% For the stopping criterion, we set $\epsilon_1=\epsilon_2 = 10^{-4}$ and $\epsilon_3=3\times 10^{-5}$.
We use a radar-only scheme that solves \eqref{eq:prob_formulation1} without the communication constraints as a baseline.
Further, we compare the proposed algorithm to the algorithm in \cite{liuDualFunctionalRadarCommunicationWaveform2021}, which optimizes the beam pattern shaping cost on a symbol-by-symbol basis rather than block-by-block, under a per-user CI constraint.
% For clarity, we refer to our proposed scheme and the CI-SLP approach in \cite{liuDualFunctionalRadarCommunicationWaveform2021} as CI-BLP and CI-SLP, respectively.

\textbf{Beam pattern synthesis:}
% Fig. \ref{fig:Beam_Pattern1} compares the synthesized beam patterns of the proposed algorithm (CI-BLP), the CI-SLP approach \cite{liuDualFunctionalRadarCommunicationWaveform2021}, and the radar-only scheme.
% It can be seen that the proposed algorithm achieves lower sidelobes than the CI-SLP approach.
% This gain comes from block-level optimization.
% The CI-SLP approach optimizes the objective for one symbol at a time. 
% In contrast, the CI-BLP approach optimizes the beam pattern for the entire block, which enables lower sidelobes.
% % This performance gain comes from block-wise optimization.
% % The CI-SLP method can be seen as a myopic approach, where the objective is optimized only for one symbol at a time.
% % As opposed to the CI-SLP method, the proposed CI-BLP optimizes the beam pattern for the entire block, which enables lower sidelobes.
% It is evident the radar-only scheme matches the ideal beam pattern most closely.
% % This is due to the trade-off between information transmission and radar sensing.
% Since no communication constraint was imposed, the radar-only scheme provides the performance bound of DFRC schemes.
Fig. \ref{fig:Beam_Pattern1} compares the synthesized beam patterns of the proposed algorithm (CI-BLP), the CI-SLP approach \cite{liuDualFunctionalRadarCommunicationWaveform2021}, and the radar-only scheme.
It is evident that the proposed CI-BLP approach achieves lower spatial sidelobes than the CI-SLP approach.
This gain comes from block-level optimization.
The CI-SLP approach optimizes the objective for one symbol at a time, which can be seen as a myopic approach. 
As opposed to the CI-SLP method, the CI-BLP approach optimizes the beam pattern for the entire block, which enables lower sidelobes.
% This performance gain comes from block-wise optimization.
% The CI-SLP method can be seen as a myopic approach, where the objective is optimized only for one symbol at a time.
% As opposed to the CI-SLP method, the proposed CI-BLP optimizes the beam pattern for the entire block, which enables lower sidelobes.
In addition, it can be observed that the radar-only scheme has the lowest spatial sidelobes.
% matches the ideal beam pattern most closely.
This is due to the trade-off between information transmission and radar sensing.
Since no communication constraint was imposed, the radar-only scheme provides the performance bound of DFRC schemes.

% \begin{figure}
%         \centering
%         {\includegraphics[width=.8\linewidth]{./Fig/fig_CC1.eps}}
%         \caption{Cross-correlation between $\theta_1=-30^{\circ}$ and $\theta_2=40^{\circ}$.}
%         \label{fig:CC1}
% \end{figure}
\begin{figure}
        \center{\includegraphics[width=.8\linewidth]{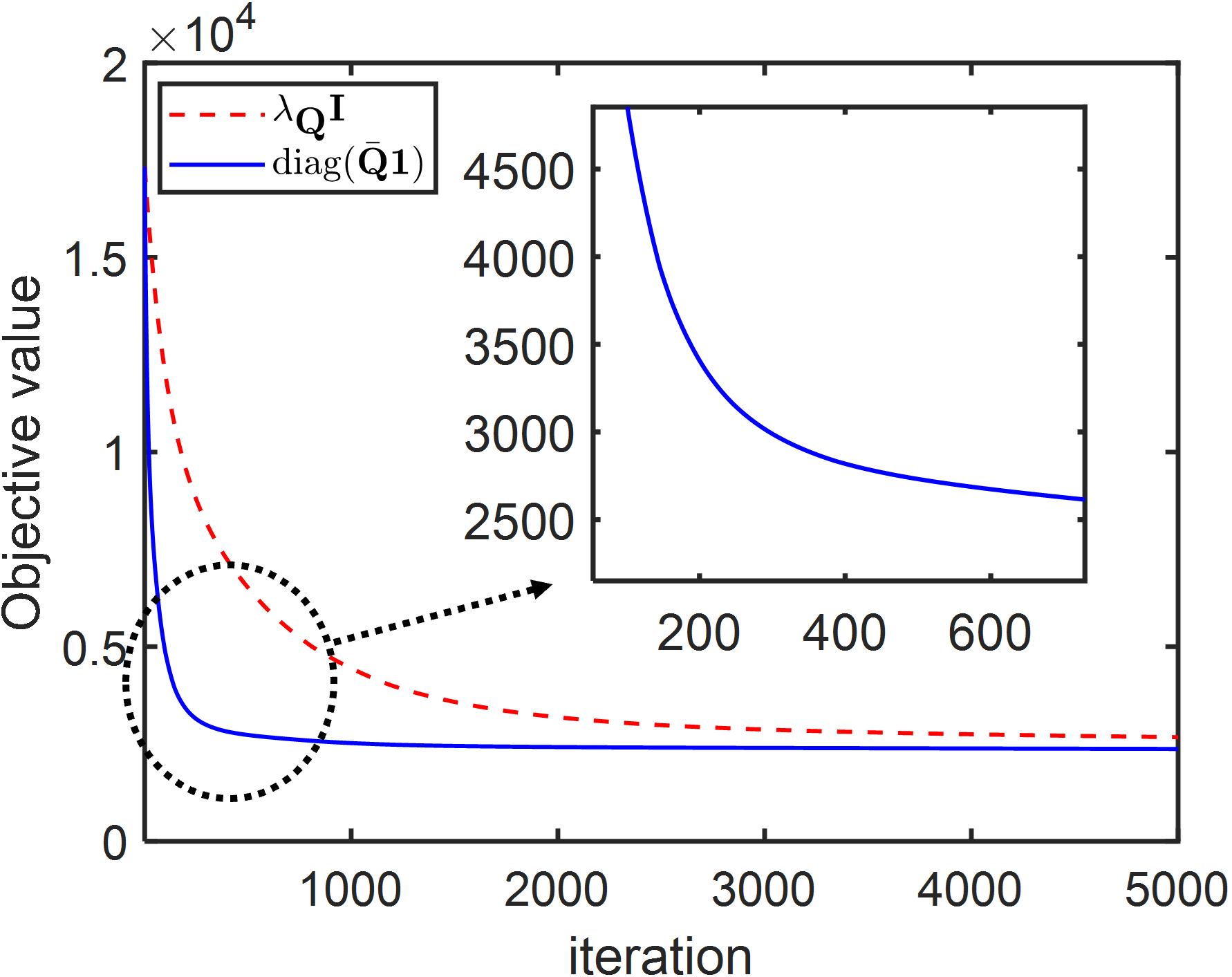}}
        \caption{ Convergence of the proposed majorizer and the maximum eigenvalue-based majorizer.}
        \label{fig:convergence}
        \vspace{-4mm}
\end{figure}

\textbf{Autocorrelation properties:} Next, we evaluate the correlation properties of waveforms designed by the proposed algorithm in comparison with CI-SLP and the radar-only scheme.
Fig. \ref{fig:AC1} and Fig. \ref{fig:AC2} plot the autocorrelations of targets $\theta_1$ and $\theta_2$, respectively.
It can be verified that the CI-BLP scheme outperforms the CI-SLP scheme in terms of autocorrelation with an approximate improvement of $\SI{10}{\decibel}$. 
The CI-SLP approach designs the waveform on a symbol-by-symbol basis, which does not address the temporal correlation between the symbols in a transmission block.
In contrast, the CI-BLP approach suppressed temporal sidelobes effectively due to block-level ISL minimization.
Consequently, the CI-BLP scheme approaches the radar-only scheme quite closely in terms of autocorrelation, leading to improvements in delay estimation accuracy.

% It can be seen that the CI-SLP scheme performs the worst in autocorrelation, and its sidelobes are approximately $\SI{10}{\decibel}$ higher than the CI-BLP scheme.
% We set $K=3$ and $(\omega_{bp},\omega_{ac},\omega_{cc})=(0.01,0.02,0.02)$.
% Fig. \ref{fig:AC1} and Fig. \ref{fig:AC2} plot the autocorrelations of each target angle.
% Likewise, the radar-only scheme shows the lowest range sidelobe levels than other schemes.
% This is due to the trade-off between sidelobe reduction and information transmission.
% It can be seen that the CI-SLP scheme performs the worst in autocorrelation, and its sidelobes are approximately $10 dB$ higher than the CI-BLP scheme.
% The CI-SLP approach designs the waveform on a symbol-by-symbol basis, which does not address the temporal correlation between the symbols in a transmission block.

\textbf{Cross-correlation properties:} 
Fig. \ref{fig:CC1} compares the space-time cross-correlation between two target angles.
Overall, the proposed algorithm is approximately $10 dB$ lower in cross-correlation than the CI-SLP approach.
For a similar reason to the autocorrelation figures, our proposed approach suppresses the cross-correlation between the targets, while the CI-SLP approach does not address correlation properties.
Owing to the lower cross-correlation, our proposed CI-BLP approach can better distinguish targets at different distances than the CI-SLP approach.

\textbf{Convergence of the proposed majorizer:} In Fig. \ref{fig:convergence}, we compare the convergence speed of the proposed algorithm with the proposed majorizer using Lemma 3 and the traditional largest eigenvalue-based majorizer.
For this simulation, we configure $L=8$ and $(\omega_{bp},\omega_{ac},\omega_{cc})=(1,0,0)$, which optimizes the beam pattern only.
Clearly, the slope of the proposed majorizer is steeper than that of the largest eigenvalue-based majorizer.
The objective value using the proposed majorizer rapidly decreases until around $600$ iterations whereas that using the largest eigenvalue-based majorizer descends much more slowly until more than $2500$ iterations.
% The objective value using the proposed majorizer rapidly decreases until around $200$ iterations and begins to converge.
% In contrast, the objective value using the largest eigenvalue-based majorizer descends at a much slower rate until more than $2500$ iterations.
This demonstrates the superiority of the proposed majorizer.

\section{Conclusion}

This paper investigated the problem of designing constant modulus waveforms for DFRC systems.
We optimized the space-time correlations of the waveform as well as its spatial beam pattern for high-resolution parameter estimation.
For the communication function, we studied the use of CI-BLP, where communication symbols are designed via block-level optimization rather than symbol-by-symbol optimization.
We solved the formulated problem using the MM technique and the Lagrange method of multipliers.
For faster convergence, we proposed a novel majorizer that outperforms traditional majorizers.
% Then we evaluate the developed algorithms through comprehensive simulations.
Simulation results showed the effectiveness of our proposed majorizer and proposed block-level approach.

\color{black}

\appendices

\section{Beam Pattern Shaping Cost Derivation}\label{sec:appendix_A}
The beam pattern shaping cost is given by
\begin{equation}  \label{eq:f_x1_vector}
\tilde{g}_{bp}(\alpha,\bm{x})=\displaystyle\sum_{u=1}^U|\alpha G_d(\theta_u)-\bm{x}^H\textbf{A}_u\bm{x}|^2.
\end{equation}
Clearly, the above pattern beam shaping cost is quadratic in $\alpha$.
The partial derivative of the beam pattern shaping cost with respect to $\alpha$ is given by
\begin{equation}    \nabla_{\alpha}\tilde{g}_{bp}(\alpha,\bm{x})=\displaystyle\sum_{u=1}^U\left(2\alpha G_d^2(\theta_u)-2\bm{x}^H\textbf{A}_u\bm{x}G_d(\theta_u)\right).
\end{equation}
The second-order partial derivative is given by
\begin{equation}    
    \nabla_{\alpha}(\nabla_{\alpha}\tilde{g}_{bp}(\alpha,\bm{x}))=\displaystyle\sum_{u=1}^M2G_d^2(\theta_u) \geq 0.
\end{equation}
Hence, the beam pattern shaping cost is convex in $\alpha$ with fixed $\bm{x}$ and the optimal $\alpha$ can be found at the critical point.
The solution $\alpha^*$ can be readily obtained as 
\begin{equation} \label{eq:alpha_sol}   \alpha^*=\frac{\displaystyle\sum_{u=1}^U\bm{x}^H\textbf{A}_u\bm{x}G_d(\theta_u)}{\displaystyle\sum_{u=1}^U G_d^2(\theta_u)}.  
\end{equation}
We can simplify the beam pattern shaping cost by plugging \eqref{eq:alpha_sol} into \eqref{eq:f_x1_vector} as
\begin{equation}
    g_{bp}(\bm{x})=\displaystyle\sum_{u=1}^U|\bm{x}^H\textbf{B}_u\bm{x}|^2,
\end{equation}
where
\begin{equation}
% (\textbf{I}_L\otimes\textbf{a}^H(\theta_u))^H(\textbf{I}_L\otimes\textbf{a}^H(\theta_u))
    \textbf{B}_u \triangleq \left(\frac{G_d(\theta_u)\displaystyle\sum_{u'=1}^U\textbf{A}_{u'}G_d(\theta_{u'})}{\displaystyle\sum_{u'=1}^UG_d^2(\theta_{u'})}-\textbf{A}_u\right).
\end{equation}
\vspace{-4mm}
\section{Space-Time Correlation Function}\label{sec:appendix_B}

The vector-form space-time correlation function can be derived using the basic properties of the trace and vectorization operators as
\begin{align*}
\chi_{u,q,q'}&=|\textbf{a}^H(\theta_q)\textbf{X}\textbf{J}_p\textbf{X}^H\textbf{a}(\theta_{q'})|^2 \\
&=|\text{Tr}\left(\textbf{X}^H\textbf{a}(\theta_{q'})\textbf{a}^H(\theta_q)\textbf{X}\textbf{J}_p\right)|^2 \\
&=|\text{Tr}\left((\textbf{a}(\theta_{q})\textbf{a}^H(\theta_{q'})\textbf{X})^H\textbf{X}\textbf{J}_p\right)|^2 \\
&=|\text{vec}^H\left(\textbf{a}(\theta_{q})\textbf{a}^H(\theta_{q'})\textbf{X}\right)\text{vec}\left(\textbf{X}\textbf{J}_p\right)|^2 \\
&=|\left((\textbf{I}_L \otimes \textbf{a}(\theta_{q})\textbf{a}^H(\theta_{q'}))\bm{x}\right)^H(\textbf{J}_{u}^T \otimes \textbf{I}_{N_T})\bm{x}|^2 \\
&=|\bm{x}^H\underbrace{\left(\textbf{J}_{-u}\otimes \textbf{a}(\theta_{q'})\textbf{a}^H(\theta_{q})\right)}_{\textbf{D}_{u,q,q'}}\bm{x}|^2 \\
&= |\bm{x}^H \textbf{D}_{u,q,q'}\bm{x}|^2.
\end{align*}

% \section{Proof of Lemma \ref{eq:lemma_nonconvexity2}} \label{sec:appendix_C}

% To show the nonconvexity of the feasible set, we transform the constraints in \eqref{eq:prob_formulation1} into a real-valued constraint as
% \begin{equation}\label{eq:prob_formulation_real}
% \begin{aligned}
%  & \bar{\textbf{h}}^T_m\bar{\bm{x}} \geq \Gamma_m, \ \forall m=1,2,\dots,2KL\\
%  &\bar{x}_n^2+\bar{x}_{n+LN_T}^2 = 1, \ \forall n=1,2,\dots,LN_T \\
% \end{aligned}
% \end{equation}
% where $\bar{\bm{x}}=[\Re\{\bm{x}^T\},\Im\{\bm{x}^T\}]^T$ and $\bar{\textbf{h}}_m=[\Re\{\tilde{\textbf{h}}^T_m\},\Im\{\tilde{\textbf{h}}^T_m\}]^T$.
% Since the communication constraints are linear inequalities, their intersection forms a polygon.
% The feasible region of the constant modulus constraint takes the shape of a unit circle in the $n$th and $(n+LN_T)$th coordinates.
% Consequently, the intersection of the feasible sets turns out to be an arc of each circle, which proves that the feasible set is nonconvex.

% \section{Proof of  \eqref{ineq:quadratic_majorizer3}}\label{sec:appendix_D}
\section{Proof of Lemma 3}\label{sec:appendix_C}

By applying Lemma 1 and Lemma 2 to \eqref{eq:obj_quartic}, we have
\begin{equation}\label{ineq:quartic_to_quadratic}
\begin{aligned}
    &\text{vec}^H(\bm{x}\bm{x}^H){\bm{\Psi}}\text{vec}(\bm{x}\bm{x}^H) 
    \\   
    &\leq \text{vec}^H(\bm{x}\bm{x}^H)\text{diag}(\bar{\bm{\Psi}}\textbf{1})\text{vec}(\bm{x}\bm{x}^H) \\  
    &\ +2\Re\{\text{vec}^H(\bm{x}\bm{x}^H)({\bm{\Psi}}-\text{diag}(\bar{\bm{\Psi}}\textbf{1}))\text{vec}(\bm{x}_t\bm{x}_t^H)\}  \\ 
    &\ +\text{vec}^H(\bm{x}_t\bm{x}_t^H)(\text{diag}(\bar{\bm{\Psi}}\textbf{1})-{\bm{\Psi}})\text{vec}(\bm{x}_t\bm{x}_t^H).
\end{aligned}
\end{equation}
    Under the strict constant modulus constraint, the first term on the right-hand side becomes a constant \cite{zhao2016unified}:
\begin{equation}\label{eq:powre_square}
\begin{aligned}
    &\text{vec}^H(\bm{x}\bm{x}^H)\text{diag}(\bar{\bm{\Psi}}\textbf{1})\text{vec}(\bm{x}\bm{x}^H)  \\ \nonumber
    &= \text{vec}^H(\bm{x}\bm{x}^H)\left(\bar{\bm{\Psi}}\textbf{1}\odot\text{vec}(\bm{x}\bm{x}^H)\right) \\ \nonumber
    &=\text{Tr}\left(\bm{x}\bm{x}^H\text{mat}(\bar{\bm{\Psi}}\textbf{1}\odot\text{vec}(\bm{x}\bm{x}^H))\right) \\ \nonumber
    &= \text{Tr}\left(\bm{x}\bm{x}^H\left(\underbrace{\text{mat}(\bar{\bm{\Psi}}\textbf{1})}_{\textbf{E}}\odot\bm{x}\bm{x}^H\right)\right) \\ \nonumber
    &= \text{Tr}\left(\bm{x}\bm{x}^H\left(\textbf{E}\odot\bm{x}\bm{x}^H\right)\right)
     \\ \nonumber
    &\overset{\mathrm{(a)}}{=} \text{Tr}\left(\left(\bm{x}\bm{x}^H\odot\bm{x}^*\bm{x}^T\right)\textbf{E}\right)
    \\ \nonumber
    &= \frac{P_T^2}{N_T^2}\textbf{1}^T\textbf{E}\textbf{1},
\end{aligned}
\end{equation}
where equality $(a)$ follows from $\text{Tr}(\textbf{A}^T(\textbf{B}\odot \textbf{C}))=\text{Tr}((\textbf{A}^T\odot \textbf{B}^T)\textbf{C})$ \cite{magnus2019matrix}. 
By simplifying constant terms, the objective can be rewritten as
\begin{equation}\label{ineq:quadratic_majorizer2}
    \begin{aligned}
    &\omega_{bp}g_{bp}(\bm{x})+\omega_{ac}g_{ac}(\bm{x})+\omega_{cc}g_{cc}(\bm{x})  \\ 
    &\leq 2\Re\{\text{vec}^H(\bm{x}\bm{x}^H)({\bm{\Psi}}-\text{diag}(\bar{\bm{\Psi}}\textbf{1})\text{vec}(\bm{x}_t\bm{x}_t^H)\}
    +const \\
    &= 2\Re\{\text{vec}^H(\bm{x}\bm{x}^H){\bm{\Psi}}\text{vec}(\bm{x}_t\bm{x}_t^H)\} \\
    &- 2\Re\{\text{vec}^H(\bm{x}\bm{x}^H)\text{diag}(\bar{\bm{\Psi}}\textbf{1})\text{vec}(\bm{x}_t\bm{x}_t^H)\}
    +const.
\end{aligned}
\end{equation}
The first term on the right-hand side of \eqref{ineq:quadratic_majorizer2} can be decomposed as
\begin{equation}
\begin{aligned}
    &2\Re\{\text{vec}^H(\bm{x}\bm{x}^H){\bm{\Psi}}\text{vec}(\bm{x}_t\bm{x}_t^H)\} \\ \nonumber
    &= 2\Re\{\text{vec}^H(\bm{x}\bm{x}^H)\left(\omega_{bp}\bm{\Psi}_1+\omega_{ac}\bm{\Psi}_2+\omega_{cc}\bm{\Psi}_3\right)\text{vec}(\bm{x}_t\bm{x}_t^H)\}
\end{aligned}    
\end{equation}
Using the basic properties of the vectorization and trace operators, we have
\begin{align*}
    &\text{vec}(\bm{x}\bm{x}^H)\bm{\Psi}_1\text{vec}^H(\bm{x}_t\bm{x}_t^H) \\
    &= \displaystyle\sum_{u=1}^U\text{vec}(\bm{x}\bm{x}^H)\text{vec}(\textbf{B}_u)\text{vec}^H(\textbf{B}_u)\text{vec}^H(\bm{x}_t\bm{x}_t^H) \\
    &=\displaystyle\sum_{u=1}^U\text{Tr}(\bm{x}\bm{x}^H\textbf{B}_u)\text{Tr}(\textbf{B}^H_u\bm{x}_t\bm{x}_t^H)\\
    &=\displaystyle\sum_{u=1}^U\bm{x}^H\textbf{B}_u\bm{x}\bm{x}_t^H\textbf{B}^H_u\bm{x}_t \\
    &=\bm{x}^H\underbrace{\left( \displaystyle\sum_{u=1}^U\bm{x}_t^H\textbf{B}^H_u\bm{x}_t\textbf{B}_u\right)}_{\bm{\Phi}_1}\bm{x}\\
    &= \bm{x}^H\bm{\Phi}_1\bm{x}.
\end{align*}
Likewise, we can rewrite the second and third terms, respectively, as
\begin{equation}
    \begin{aligned}
        &\text{vec}^H(\bm{x}\bm{x}^H)\bm{\Psi}_2\text{vec}(\bm{x}_t\bm{x}_t^H)\\ &= \bm{x}^H\underbrace{\left(\displaystyle\sum_{q=1}^Q\displaystyle\sum_{\substack{\tau=-P+1, \\ \tau\neq 0}}^{P-1}\bm{x}_t^H\textbf{D}^H_{\tau,q,q}\bm{x}_t\textbf{D}_{\tau,q,q}\right)}_{\bm{\Phi}_2} \bm{x}\\
    &= \bm{x}^H\bm{\Phi}_2\bm{x},\\
    &\text{vec}^H(\bm{x}\bm{x}^H)\bm{\Psi}_3\text{vec}(\bm{x}_t\bm{x}_t^H) \\&= \bm{x}^H\underbrace{\left(\displaystyle\sum_{q=1}^{Q}\displaystyle\sum_{\substack{q'=1, \\ q'\neq q}}^Q\displaystyle\sum_{\tau=-P+1}^{P-1}\bm{x}_t^H\textbf{D}^H_{\tau,q,q'}\bm{x}_t\textbf{D}_{\tau,q,q'}\right)}_{\bm{\Phi}_3} \bm{x}\\
    &= \bm{x}^H\bm{\Phi}_3\bm{x}.
    \end{aligned}
\end{equation}
The second term on the right-hand side of \eqref{ineq:quadratic_majorizer2} can be rewritten as
\begin{equation}
    \begin{aligned}
        &\text{vec}(\bm{x}\bm{x}^H)^H\text{diag}(\bar{\bm{\Psi}}\textbf{1})\text{vec}(\bm{x}_t\bm{x}_t^H)  \\ \nonumber
    &= \text{vec}(\bm{x}\bm{x}^H)^H\left(\bar{\bm{\Psi}}\textbf{1}\odot\text{vec}(\bm{x}_t\bm{x}_t^H)\right) \\ \nonumber
    &=\text{Tr}\left(\bm{x}\bm{x}^H\text{mat}(\bar{\bm{\Psi}}\textbf{1}\odot\text{vec}(\bm{x}_t\bm{x}_t^H))\right) \\ \nonumber
    &= \text{Tr}\left(\bm{x}\bm{x}^H\left(\underbrace{\text{mat}(\bar{\bm{\Psi}}\textbf{1})}_{\textbf{E}}\odot\bm{x}_t\bm{x}_t^H\right)\right) \\ \nonumber
    &= \bm{x}^H\left(\textbf{E}\odot\bm{x}_t\bm{x}_t^H\right)\bm{x}.
    \end{aligned}
\end{equation}
Based on the above discussions, \eqref{ineq:quadratic_majorizer2} can be simplified as 
\begin{equation}\label{ineq:quadratic_majorizer3}
    \begin{aligned}
        &\omega_{bp}g_{bp}(\bm{x})+\omega_{ac}g_{ac}(\bm{x})+\omega_{cc}g_{cc}(\bm{x}) \\
        % &\leq 2\Re\{\text{vec}^H\left(\bm{x}\bm{x}^H)({\bm{\Psi}}-\text{diag}(\bar{\bm{\Psi}}\textbf{1})\right)\text{vec}(\bm{x}_t\bm{x}_t^H)\} +const
        %  \\ 
         &\leq \bm{x}^H\underbrace{2\left(\omega_{bp}\bm{\Phi}_1+\omega_{ac}\bm{\Phi}_2+\omega_{cc}\bm{\Phi}_3-(\textbf{E}\odot \bm{x}_t\bm{x}_t^H)\right)}_{\bm{\Phi}}\bm{x} +const\\
         &= \bm{x}^H\bm{\Phi}\bm{x}+const.
    \end{aligned}    
\end{equation}

\bibliographystyle{IEEEtran}
\bibliography{IEEEabrv,references}

\end{document}